\documentclass[AMA,Times1COL]{WileyNJDv5} 

\articletype{Article Type}%

\received{Date Month Year}
\revised{Date Month Year}
\accepted{Date Month Year}
\journal{Journal}
\volume{00}
\copyyear{2023}
\startpage{1}
\usepackage{algorithm}
\usepackage{algpseudocode}  
\usepackage{amsmath}
\usepackage{amssymb}
\usepackage{makecell}
\usepackage{spreadtab}
\usepackage{booktabs}
\usepackage{array}
\begin{document}

\title{Lightweight Fault Detection Architecture for NTT on FPGA}

\author[1]{Rourab Paul}

\author[2,3]{Paresh Baidya}

\author[3]{Krishnendu Guha}

\authormark{TAYLOR \textsc{et al.}}
\titlemark{PLEASE INSERT YOUR ARTICLE TITLE HERE}

\address[1]{\orgdiv{Computer Science and Engineering}, \orgname{Shiv Nadar University Chennai}, \orgaddress{\state{Tamil Nadu}, \country{India}}}

\address[2]{\orgdiv{Department of Mathematics}, \orgname{National Institute of Technology Jamshedpur}, \orgaddress{\state{} \country{India}}}
\address[3]{\orgdiv{Computer Science Engineering}, \orgname{SOA University, Bhubaneswar}, \orgaddress{\state{Odisha}, \country{India}}}

\address[4]{\orgdiv{Computer Science and Information Technology}, \orgname{University College Cork}, \orgaddress{\state{} \country{Ireland}}}

\corres{Corresponding author Rourab Paul, This is sample corresponding address. \email{rourabpaul@snuchennai.edu.in}}



\abstract[Abstract]{
Post-Quantum Cryptographic (PQC) algorithms are mathematically secure and resistant to quantum attacks but can still leak sensitive information in hardware implementations due to natural faults or intentional fault injections. The intent fault injection in side-channel attacks reduces the reliability of crypto implementation in future generation network security procesors. In this regard, this research proposes a lightweight, efficient, recomputation-based fault detection module implemented on a Field Programmable Gate Array (FPGA) for Number Theoretic Transform (NTT). The NTT is primarily composed of memory units and the Cooley–Tukey Butterfly Unit (CT-BU), a critical and computationally intensive hardware component essential for polynomial multiplication. NTT and polynomial multiplication are fundamental building blocks in many PQC algorithms, including Kyber, NTRU, Ring-LWE, and others. 
In this paper, we present a fault detection method called :  Recomputation with a Modular Offset (\texttt{REMO}) for the logic blocks of the CT-BU using Montgomery Reduction and another method called Memory Rule Checkers for the memory components used within the NTT.
The proposed fault detection framework sets a new benchmark by achieving high efficiency with significant low implementation cost. It occupies only 16 slices and a single DSP block, with a power consumption of just $3\text{mW}$ in $Artix$-$7$ FPGA. The \texttt{REMO}-based detection mechanism achieves a fault coverage of 87.2\% to 100\%, adaptable across various word sizes ($w$), fault bit counts ($\eta$), and fault injection modes. Similarly, the Memory Rule Checkers demonstrate robust performance, achieving 50.7\% to 100\% fault detection depending on $\eta$ and the nature of injected faults.
}


\jnlcitation{\cname{%
\author{Rourab Paul$^1$},
\author{Paresh Baidya$^{23}$},
 and
\author{Krishnendu Guha$^4$}}.
\ctitle{On simplifying ‘incremental remap’-based transport schemes.} \cjournal{\it J Comput Phys.} \cvol{2021;00(00):1--18}.}

\maketitle

\renewcommand\thefootnote{}

\renewcommand\thefootnote{\fnsymbol{footnote}}
\setcounter{footnote}{1}
\section{Introduction}

Faults in PQC hardware can occur naturally due to shrinking device dimensions, which increase the probability of errors caused by ion interference or electromagnetic radiation. These densely packed devices are particularly susceptible to internal faults induced by ion beam radiation in the Configurable Logic Block (CLB) of FPGAs \cite{niko}. Additionally, in side-channel attacks, adversaries deliberately introduce faults into PQC hardware and analyze variations in physical parameters, such as power consumption and execution time, to extract sensitive information.
Therefore, mathematically secure PQC algorithms may still be vulnerable in hardware implementations to various side-channel attacks, such as power analysis, timing analysis, electromagnetic (EM) attacks and fault injection attacks. Specifically, fault injection attacks \cite{ravi} \cite{keita} exploit deliberate faults introduced into PQC hardware to expose intermediate states or reveal secret data. Barenghi et al. \cite{barenghi} discussed various cost-effective methods for injecting faults into existing cryptographic systems. Recently, Primas et al. \cite{primas} and Prasanna et al. \cite{ravi3} conducted a Soft Analytical Side-Channel Attack (SASCA) on the NTT, employing a probabilistic model utilizing power and timing data.
\par
Preventing fault occurrences in PQC hardware provides a robust defence against both natural and intentional faults.
To address various fault detection strategies, several fault detection approaches have been proposed in the literature. We categorize the existing work into two groups: (i) Fault detection solutions for various components of crypto algorithms, and (ii) Fault detection specifically targeting the NTT (Number Theoretic Transform).

\subsection{Fault Detection in various Crypto Algorithms [\cite{saeed}  \cite{canto} \cite{howe} \cite{canto2} \cite{kermani} \cite{cintas} \cite{kamal} \cite{ahmadECSM} \cite{ahmadiNAF} \cite{dominguez}]}
The modular exponentiation is a crucial operation in cryptography.
Saeed et.al \cite{saeed} proposed a  recomputation based fault detection model for modular exponent $x^ymod~n$ implemented on an ARM Cortex-A72 processors, AMD/Xilinx Zynq Ultrascale+ and Artix-7 FPGA. Their approach involves recomputing the modular exponent using encoded values of $x$ and $y$. The proposed method achieves near 100\% error detection accuracy with approximately 7\% computational overhead and less than 1\% area overhead compared to unprotected architectures.
 Saeed et.al \cite{saeed} employed the popular Right-to-Left Exponentiation algorithm \cite{menezes}, where the number of iterations depends on the number of `$1$'s in exponent $y$. This conditional operation makes the algorithm vulnerable \cite{kocher} to timing analysis attacks. 
 Canto et al. \cite{canto} presented various fault detection mechanisms design for finite-field operations, including addition, subtraction, multiplication, squaring, and inversion, specifically within the framework of the code-based McEliece cryptosystem. The proposed approaches used different error detection techniques, such as regular parity, interleaved parity, CRC-2, and CRC-8, to enhance the fault detection capabilities. These methods are applied to different elements of the Key Generator, focusing on improving error detection accuracy in operations such as multiplications and inversions in the finite field $GF(2^{13})$.
 Howe et al. \cite{howe} designed a fault detection module for error sampler used in lattice based cryptography. They introduced three methods: low-cost test, standard test, and expensive test designed for FPGA implementation to verify whether the output distribution of the error sampler matches its expected Gaussian or binomial shape.
 The study in \cite{canto2} proposed three fault detection models for  Multiply and Accumulate (MAC) unit of lattice-based Key Encapsulation Mechanisms (KEMs) and applied them to hardware accelerators of three NIST PQC finalists: FrodoKEM, SABER, and NTRU. The proposed schemes are based on the recomputation technique such as Recomputing with Shifted Operands (RESO), Recomputing with Negated Operands (RENO) and  Recomputing
with Scaled operands (RECO), implemented on a Kintex UltraScale+ device. Their implementation on FPGA devices demonstrates minimal overhead and significant error detection coverage, ensuring compatibility with other cryptographic systems utilizes hardware accelerators.
 Kermani et al. \cite{kermani} introduced a novel error detection scheme for Galois Counter Mode (GCM) implemented on a 65nm Application Specific Integrated Circuit (ASIC) platform, specifically developed to improve data integrity verification. The proposed methodology improved compatibility with various block ciphers and finite field multipliers by employing a technique Re-computation of Swapped Ciphertext and Additional Authenticated Blocks (RESCAB). In this work, the primary computational unit, Galois Hash (GHASH) computed over the finite field $GF(2^{128})$, while the RESCAB module processed swapped inputs concurrently within another $GF(2^{128})$ instance. Then, fault detection achieved by comparing the outputs of GHASH and RESCAB. The architecture in \cite{kermani}, demonstrated significant design flexibility and reliability, as validated through hardware implementations and fault simulation analyzes.
  Cintas et al. \cite{cintas} proposed an error detection mechanism for Goppa arithmetic units used in the McEliece cryptosystem. They used the algebraic structure of composite fields in this cryptosystem. Their approach included implementing a Parity Checker for various sub-blocks of McEliece. The proposed methods in \cite{cintas} were not limited to arithmetic units but were also suitable for core functions of other public-key cryptosystems that relied on composite fields as their mathematical foundation. Additionally, the authors presented FPGA-based implementations of Goppa polynomial evaluation (GPE) and analyzed the performance overhead for different configurations.
  In \cite{kamal}, the authors investigated various techniques to strengthen the resilience of NTRUEncrypt hardware implementations, against fault analysis attacks. They used the cipher's algebraic properties, and proposed countermeasures based on error detection codes and spatial/temporal redundancies. A detailed evaluation of these methods was provided by comparing their error detection efficiency along with their impact on decryption throughput and hardware area.
   Ahmadi et al. \cite{ahmadECSM} proposed a fault detection scheme for the window method in elliptic curve scalar multiplication (ECSM). They introduced refined algorithms and hardware implementations to address both permanent and transient errors. Using simulation-based fault injection, the schemes achieved extensive error coverage with less than 3\% clock cycle overhead on Cortex-A72 processors and a 2\% area increase on FPGAs. These results demonstrated efficient error detection with minimal resource overhead.
Ahmadi et al. \cite{ahmadiNAF} addressed a research gap in fault detection for $\tau$NAF conversion and Koblitz curve cryptosystems. They proposed an algorithm-level fault detection scheme for the single $\tau$NAF conversion algorithm and developed two additional fault detection schemes for the double $\tau$NAF conversion algorithm. The feasibility of these methods was evaluated through implementation on ARMv7 and ARMv8 architectures.
This paper \cite{dominguez} introduced fault-tolerant and error-detecting structures for elliptic curve scalar multiplication (ECSM). The proposed methods utilized recomputation, parallel computation, and encoding-decoding schemes to enhance error detection during ECSM operations. These schemes used scalar and point randomization techniques and were implemented on Xilinx Virtex 2000E FPGA. This work achieved a high error coverage with minimal computational overhead. 

\subsection{Fault Detection for NTT [\cite{khan} \cite{sven} \cite{sarker} \cite{jati} \cite{ravi3} \cite{rafael}]}
 NTT is used as both forward and inverse transforms and is one of the most widely used components in many post-quantum cryptography (PQC) algorithms.
 The recent NIST standardization of Kyber as Module Lattice-based Key Encapsulation Mechanism (ML-KEM) and Dilithium as Module Lattice-based Digital Signature Algorithm (ML-DSA) includes the Number Theoretic Transform (NTT) and KECCAK as the most critical hardware components. The NTT block consists of three main components: (i) Random Access Memories (RAMs) to store polynomial coefficients, error polynomials and secret vectors, (ii) Read-Only Memories (ROMs) to store constants such as twiddle factors, and (iii) a CT-BU to perform the core computations of the NTT. 
 Khan et al. \cite{khan} implemented a fault-tolerant memory modules (RAMs and ROMs) for storing polynomial coefficients and twiddle factors of the NTT using Hamming codes and parity bits, resulting in a 29.2\% overhead on a Virtex-7 FPGA. They implemented four fault-tolerant variants for the memory modules of the NTT used in Kyber: (i) Area Optimized (A-O) using Hamming Code, (ii) Run-Time Optimized (R-O) using Hamming Code, (iii) Area Optimized (R-O) using Hamming Code combined with Parity and (iv) Run-Time Optimized (R-O) using Hamming Code combined with Parity.
Sven et al. \cite{sven} implemented a fault-tolerant model for the Number Theoretic Transform (NTT) that can resist faults in multiplication with a twiddle factor and the addition in a butterfly operation on the ARM Cortex-M4 platform. They used interpolation and, evaluation and inverse NTT method to detect the faults inside the CT-BU. They have taken a polynomial coefficient $f(X)$, and its inverse NTT is $f(w^j)$. If the Eq. \ref{eq:int} is satisfied, it indicates that no fault occurred during the NTT computation.
\begin{align}
f(u)&=evaluation(NTT^{-1}(f(w^j))) =Interpolation(f(w^j)) 
\label{eq:int}
\end{align}
Polynomial Evaluation and Interpolation is a method that evaluates a polynomial at powers of a primitive root of unity during the forward NTT and recovers the original polynomial coefficients through interpolation during the inverse NTT, all over a finite field. However, the interpolation and evaluation processes introduce additional 3160 and 2979 clock cycles, respectively, resulting in a significant 72\% timing overhead for Dilithium.

Sarker et al. \cite{sarker} implemented a fault detection scheme for the HW/SW co-design of NTT on Spartan-6 and Zynq UltraScale+ platforms, resulting in 12.74\% resource overhead and approximately 20\% power overhead. They used recomputing with negative operands (RENO). Depending on the placement of the RENO block in the logic path, Sarker et al. \cite{sarker} proposed three versions of the NTT. 
To the best of our knowledge, no existing fault detection solution in the literature has addressed both the logic unit of the NTT (CT-BU) and the memory units separately.
\par
Apart from the fault detection schemes proposed in \cite{khan}, \cite{sven}, and \cite{sarker}, there are a few NTT implementations that incorporate certain precautions to make side-channel attacks on the NTT more difficult. 
Jati et al. \cite{jati} implemented a side channel attack protected configurable Kyber processor on Artix-7 FPGA. Different components of the configurable Kyber processor employ different fault detection methodologies. For example, alongside the original state machine, a duplicated inverted-logic state machine was used to verify the control flow integrity of the Kyber processor.  Specifically for the NTT operation in the Kyber processor, they employed a randomized memory addressing technique. Instead of using linear increments for the memory addresses of the RAMs and ROMs in the CT-BU, the addresses were randomized. This randomization of memory access in the CT-BU helps decorrelate the relationship between power consumption and NTT iterations, thereby significantly improving resistance against side-channel attacks on the NTT.
Rafael et al. \cite{rafael} proposed a locally masked NTT scheme on Artix-7 FPGA in which the input and output are masked with random twiddle factors. This approach effectively prevents the leakage of computational patterns, thereby enhancing resistance against side-channel attacks. Ravi et al. \cite{ravi2} implemented both local masking and memory access randomization techniques in their NTT design on the ARM Cortex-M4 platform to enhance resistance against side-channel attacks.

\par 
To the best of our knowledge, existing literature lacks any lightweight fault detection scheme that independently targets both the logic component (CT-BU) and the memory subsystem of the NTT.
 In this paper we propose two fault detection schemes for CT-BU and Memory Units (RAMs and ROMs) of the NTT core.
Our fault detection scheme for the CT-BU targets Montgomery reduction unit. It is based on the RECO method but is neither RESO nor RENO; rather, it is a recomputation with a modular offset (\texttt{REMO}).
 It is important to note that the proposed \texttt{REMO} based fault-tolerant model can serve as a generic fault-resistant method for any Montgomery multiplication and reduction. To the best of our knowledge, the \texttt{REMO}-based fault detection in the CT-BU of the NTT is the first of its kind. We also propose a memory address rule checker for the RAMs and ROMs used in the NTT to detect any ambiguities in memory addressing.
 The contribution of the paper can be summerized as:
\begin{itemize}
	\item This paper introduces a novel fault detection technique based on a modified word-wise Montgomery reduction algorithm, targeted for the Cooley-Tukey Butterfly Unit (CT-BU) in the NTT architecture. Unlike existing methods, the proposed \texttt{REMO} scheme integrates fault detection into the core arithmetic logic, offering protection without requiring significant implementation cost. The customizable word size ($w$) provides significant flexibility in usage of slices without compromising the fault detection efficiency. Although it is designed for the CT-BU, the versatility of this approach extends seamlessly to any hardware architecture performing modular polynomial multiplication using Montgomery reduction,
	\item This work also introduces the Memory Rule Checker (\texttt{Memory RC}), a lightweight and effective mechanism to ensure structural integrity in memory access patterns. Unlike traditional memory protection techniques that focus solely on content integrity, the proposed \texttt{Memory RC} monitors and validates the correctness of address sequences during all read/write operations of polynomial coefficients (Generate Matrix for Kyber), error vectors, secret vectors, etc in RAM and twiddle factor retrievals from ROM, across both forward and inverse NTT executions. To the best of our knowledge, this is the first approach that systematically targets address-level faults in memory-intensive cryptographic datapaths.
	\item The proposed lightweight fault detection framework introduces negligible implementation overhead across leading post quantum cryptographic schemes, including Kyber, Crystals-Dilithium, Falcon, and NTRU. Uniquely designed for modular arithmetic and memory bound operations, this model delivers high fault detection accuracy, even under diverse fault scenarios and varying numbers of corrupted bits. This is a robust protection with minimal overhead across multiple PQC standards in resource-constrained FPGA deployments.
\end{itemize}
The organization of the article is as follows:
Sections \ref{sec:remo} and \ref{sec:mem} provide detailed descriptions of the proposed fault detection methods: \texttt{REMO} for the CT-BU and the Memory Rule Checker for the memory units, respectively. The detailed hardware architecture of the NTT, including the CT-BU, RAMs, and ROMs, along with the \texttt{REMO} and Memory Rule Checker, is presented in Section \ref{sec:impl}, while the results are discussed in Section \ref{sec:res}. Finally, the conclusions are provided in Section \ref{sec:con}.

\section{Fault Detection in CT-BU}
\label{sec:remo}
The CT-BU is the most computationally intensive hardware block in any polynomial multiplication, as well as both forward and reverse NTT operations.
As shown in Equ. \ref{equ:u} and Equ. \ref{equ:v}, the CT-BU involves two major consecutive steps shown separately in line \ref{uplusv} and \ref{uminusv} of Algorithm \ref{algo:intt}.
\begin{equation}
\label{equ:u}
U=\alpha[j+k] 
\end{equation}
\begin{equation}
\label{equ:v}
V=\alpha[j+k+\frac{i}{2}] \times \omega~mod~q
\end{equation}
Here the $U$ and $V$ are used to calculate the coefficients of inverse NTT $\overline{\alpha}$, $\omega$ is a primitive $n^{th}$ root of unity and $q$ is a prime number. The calculation of $V$ involves a modular reduction on the product of a polynomial coefficient $\alpha$ and the primitive $n^{th}$ root of unity $\omega$, which is the most hardware-intensive and latency-critical operation in the NTT transformation. We use popular Montgomery reduction technique which allows an efficient hardware implementation of modular multiplication without computing modular reduction operation on the product of $\alpha$ and $\omega$. However, the Montgomery reduction operation remains the most hardware-intensive process in the CT-BU and NTT transformation. We have chosen Montgomery reduction over Barrett reduction because Barrett requires more intermediate computations than Montgomery in word-wise form on FPGA. Muller et al. \cite{muller} show that Montgomery reduction often outperforms Barrett in FPGA implementations, especially for word-wise operations.

\subsection{REMO Method}
\label{sec:method}
In our fault detection method, rather than computing the Montgomery reduction on all $l$ bits of the polynomial coefficient $\alpha(x)$ at once, we operate on smaller, fixed word sizes of $w$ bits from the total $l$ bits of $\alpha(x)$ where $w \leq l$. This modification of Montgomery reduction is required for three reasons. 
\begin{itemize}
\item The partial recomputation technique for fault detection requires intermediate data.
\item Lower value of $w$ reduces the overhead of hardware resources as $w\leq l$. 
\item The tunable $w$ can adjust the speed, resource usage, and power consumption of the design, offering architectural flexibility in optimization.
\end{itemize}

\begin{algorithm}
  \caption{Iterative NTT Algorithm}
  \label{algo:intt}
  \begin{algorithmic}[1]
    \State \textbf{Input:} $\alpha(x)$, $\omega$, $q$
    \State \textbf{Output:} $\overline{\alpha}(x)$
    \For{$i = 2$ to $l$ by $2 \times i$}
      \For{$j = 0$ to $i/2 - 1$}
        \For{$k = 0$ to $n - 1$ step $i$}
          \State $\omega_i = \omega[2^{(i-1)k}]$
          \State $U = \alpha[j+k]$  \label{line:u}
          \State $V = \text{MMRFD}(\alpha[j+k+i/2], \omega, q)$ \label{line:mmrfd}
          \State $\overline{\alpha}[j+k] = U + V$ \label{uplusv}
          \State $\overline{\alpha}[j+k+i/2] = U - V$ \label{uminusv}
        \EndFor
        \State $\omega = \omega \cdot \omega_i$
      \EndFor
    \EndFor
    \State \Return $\overline{\alpha}$
  \end{algorithmic}
\end{algorithm}

\subsection{Modified Montgomery Reduction with REMO}
The proposed fault detection method for CT-BU is recomputation based technique where encoding inputs $\alpha$ of Montgomery reduction can detect permanent and transient errors.
This paper proposes a word-wise modified Montgomery reduction that takes two operands: the multiplicand $\alpha$, the multiplier $\beta$, a modular base $q$, and the modular inverse $q'$. Here $\alpha=\alpha[j+k+\frac{i}{2}]$ and $\beta=\omega$. For the hardware implementation of Montgomery reduction, the number of bits $l$ in $\alpha$, $\beta$, $q$, and $q'$ must be divisible by the word size $w$. As shown in Algorithm \ref{algo:mmrfd}, lines \ref{line:padalpha} and \ref{line:padbeta} pad $w-p$ zeros to $\alpha$ and $\beta$ to ensure that $l$ is divisible by $w$. As shown in lines  \ref{line:alphadash} and  \ref{line:betadash}, the zero-padded $\alpha$ and $\beta$ are then stored in $\alpha'$ and $\beta'$, respectively. This modified Montgomery reduction involves two primary computations of $\mu_i$ and $\gamma_i$. Here $\mu_i$ depends on $\gamma_i$, $aw_i$ and $\beta'$; $\gamma_i$ depends on $\mu_i$, $aw_i$ and $\beta'$ where $aw_i=(\alpha'_{iw+w-1}..\alpha'_{iw})$. The proposed fault detection method, which relies on recomputation, calculates additional values $\mu_i^f$ and $\gamma_i^f$ using an encoded form of $aw_i$, denoted $aw_i^f$. As shown in modified Montgomery Reduction for Fault Detection (MMRFD) Algorithm \ref{algo:mmrfd}, line \ref{line:awif}, the value of $aw^f_i$ is $(\alpha'_{iw+w-1}..\alpha'_{iw})+K.q$. If there is no fault in $\alpha[j+k+\frac{i}{2}]$ or in $\omega$, then $\gamma_i^f$ with $aw_i^f$ and $\gamma_i$ with $aw_i$ will be the same, where $\beta = \omega$. Conventional Montgomery reduction on $l$ bits calculates $\alpha.\beta.R^{-1} \mod q$. In encoded form it calculates ($\alpha + k. q).\beta.R^{-1} \mod q$. Both these result must be the same as $k. q.\beta.R^{-1} \mod q$ term will be canceled out as $k$ and $R^{-1}$ are integers. However, the proposed modified word-wise Montgomery reduction operates on $w$ bits instead of $l$ bits in each loop iteration. This modified word-wise approach still computes the same intermediate values in $\gamma$ for both the encoded and non-encoded forms.

\begin{algorithm}[!htb]
\caption{Modified Montgomery Reduction for Fault Detection in Hardware with REMO: MMRFD($\alpha$, $\beta$, $q$)}
\label{algo:mmrfd}
\begin{algorithmic}[1]
\State \textbf{Input:} $\alpha(x) = (\alpha_{l-1}, \ldots, \alpha_0)$, $\beta(x) = (\beta_{l-1}, \ldots, \beta_0)$, $q = (q_{l-1}, \ldots, q_0)$
\State \hspace{1cm} with $R = w^l$, $\gcd(q, w)$
\State \textbf{Output:} $\gamma, f$

\State $p \gets l \bmod w$
\If{$p \neq 0$}
  \State $\alpha' \gets \text{Pad with } (w - p)$ zeros $\parallel \alpha(x)$ \label{line:padalpha}
  \State $\beta' \gets \text{Pad with } (w - p)$ zeros $\parallel \beta(x)$  \label{line:padbeta}
\Else
  \State $\alpha' \gets \alpha(x)$ \label{line:alphadash}
  \State $\beta' \gets \beta(x)$  \label{line:betadash}
\EndIf

\State $\gamma \gets (0, \ldots, 0)$
\State $\gamma^f \gets (0, \ldots, 0)$
\State $\mu \gets (0, \ldots, 0)$
\State $f \gets (0, \ldots, 0)$

\For{$i = 0$ to $(l + p)/w - 1$}
  \State $aw_i \gets (\alpha'_{iw+w-1}, \ldots, \alpha'_{iw})$
  \State $\mu_i \gets \left((\gamma_{w-1}, \ldots, \gamma_0) + aw_i \cdot (\beta'_{w-1}, \ldots, \beta'_0)\right) \cdot q' \bmod 2^w$ \label{line:mui}
  \State $\gamma_i \gets \left(\gamma_i + aw_i \cdot \beta + \mu_i \cdot q\right) / 2^w$  \label{line:gammai}

  \State $aw^f_i \gets (\alpha'_{iw+w-1}, \ldots, \alpha'_{iw}) + K \cdot q$ \label{line:awif}
  \State $\mu_i^f \gets \left((\gamma^f_{w-1}, \ldots, \gamma^f_0) + aw^f_i \cdot (\beta'_{w-1}, \ldots, \beta'_0)\right) \cdot q' \bmod 2^w$ \label{line:muif}
  \State $\gamma_i^f \gets \left(\gamma_i^f + aw^f_i \cdot \beta + \mu_i^f \cdot q\right) / 2^w$ \label{line:gammaif}

  \If{$\gamma_i \neq \gamma_i^f$}
    \State $f_i \gets 1$
  \Else
    \State $f_i \gets 0$
  \EndIf
\EndFor

\State \Return $\gamma, f$
\end{algorithmic}
\end{algorithm}

\begin{lemma}
If we divide $l$ bit $\alpha'$ in $w$ bit word-wise (segments), each word of $\alpha'$ can be expressed as: $aw_i=(\alpha'_{iw+w-1}..\alpha'_{iw})$ and the encoded  word of $\alpha'$is represented as: $aw^f_i=(\alpha'_{iw+w-1}..\alpha'_{iw})+K.q$ where $K$ is a constant and 
$q$ is the modulus. Then  $\gamma_i$ and $\gamma_i^f$ computed in the $i^{th}$ loop from  $aw_i$ and  $aw^f_i$ respectively must be same. 
\end{lemma}

\begin{proof}
From algorithm \ref{algo:mmrfd}, line \ref{line:mui} and \ref{line:gammai} We have
\begin{align*}
\mu_i&=[(\gamma_{w-1},..,\gamma_0)+aw_i.(\beta'_{w-1},..,\beta'_0)].q'~\%~2^w\\
\gamma_i&=[\gamma_i+aw_i(\beta'_{l-1},..,\beta'_0)+\mu_iq]/2^w
\end{align*}
Now replace the $aw_i$ by $aw_i^f$
\begin{align*}
\mu_i&=[(\gamma_{w-1},..,\gamma_0)+\underbrace{(aw_i+K.q)}_{aw^f_i}.(\beta'_{w-1},..,\beta'_0)].q'~\%~2^w\\
&=[(\gamma_{w-1},..,\gamma_0)+aw_i.(\beta'_{w-1},..,\beta'_0)+k.q.(\beta'_{w-1},..,\beta'_0)]q' \% 2^w  \\
&=[(\gamma_{w-1},..,\gamma_0)+aw_i.(\beta'_{w-1},..,\beta'_0)]q'+k.\overbrace{q.q'}^{-1}.(\beta'_{w-1},..,\beta'_0) \% 2^w \\ 
&=[(\gamma_{w-1},..,\gamma_0)+aw_i.(\beta'_{w-1},..,\beta'_0)]q'
-k.(\beta'_{w-1},..,\beta'_0) \% 2^w 
\end{align*}
Therefore,
\begin{align*}
 \mu_i+k.(\beta'_{w-1},..,\beta'_0)&=[(\gamma_{w-1},..,\gamma_0)+aw_i.(\beta'_{w-1},..,\beta'_0)]q' \% 2^w 
 \end{align*}
Now,
 \begin{align*}
\gamma_i^f&=[\gamma_i^f+aw_i.(\beta'_{l-1},..,\beta'_0)+\mu_i.q]/2^w\\
&=[\gamma_i+\underbrace{(aw_i+K.q)}_{aw^f_i}.(\beta'_{l-1},..,\beta'_0)+ \underbrace{[(\gamma_{w-1},.,\gamma_0)+aw_i.(\beta'_{w-1},.,\beta'_0)]q'-k.(\beta'_{w-1},.,\beta'_0)]}_{u_i}.q]/2^w\\
&=[\gamma_i+aw_i.(\beta'_{l-1},..,\beta'_0)+\mu_i.q]/2^w+ q.\underbrace{\frac{k}{2^w}.[(\beta'_{l-1},..,\beta'_0)-(\beta'_{w-1},..,\beta'_0)]}_{t}\\
&=[\gamma_i+aw_i.(\beta'_{l-1},..,\beta'_0)+\mu_i.q]/2^w+q.t\\
&=\gamma_i+q.t
\end{align*}
As $t$ is an integer, after applying the Montgomery transformation, the final value of $\gamma_i^f$ is:
 \begin{align*}
 \gamma_i^f&= \gamma_i^f.R~\%~q= (\gamma_i+q.t).R~\%~q=\gamma_i
\end{align*}
\end{proof}
\subsection{Hardware Architecture of REMO: $\gamma_i^f$ \& $\gamma_i$}
The values of $\gamma_i$ and $\gamma_i^f$ shown in lines \ref{line:gammai} and \ref{line:gammaif} of Algorithm \ref{algo:mmrfd} are computed by two hardware blocks named $\gamma_i$ $Gen$ and $\gamma_i^f$ \texttt{REMO}, respectively. The $\gamma_i$ $Gen$ uses three multipliers ($\times$) and two adders ($+$), whereas the $\gamma_i^f$ \texttt{REMO} requires four multipliers ($\times$) and three adders ($+$). The $\gamma_i$ $Gen$ and $\gamma_i^f$ \texttt{REMO} use one right shifter to shift $\alpha'$ in each clock cycle, generating word-wise values $aw_i$. Here $\omega'$ is $\beta'$ and $\omega'_0$ $\beta'_{w-1},..,\beta'_0$ of Algorithm \ref{algo:mmrfd}. Modulus and division operations are computationally expensive on FPGA hardware. Therefore, the modulus operations: $\mod~2^w$ in Algorithm~\ref{algo:mmrfd}, as shown in line~\ref{line:mui} and line~\ref{line:muif}, are implemented by retaining only the least significant $w$ bits of $\mu_i$ and $\mu_i^f$. On the other hand, the division operations: /$2^w$ in Algorithm~\ref{algo:mmrfd}, as shown in line~\ref{line:gammai} and line~\ref{line:gammaif}, are implemented by discarding the least significant $w$ bits of $\gamma_i$ and $\gamma_i^f$. These two approaches make $MMRFD$ significantly lightweight. The Fig \ref{fig:remo} shows the details hardware architecture of $\gamma_i$ $Gen$ and \texttt{REMO}: $\gamma_i^f$. 
\begin{figure*}[!htb]
\centering
\vspace{-5pt}
\includegraphics[width=0.8\textwidth]{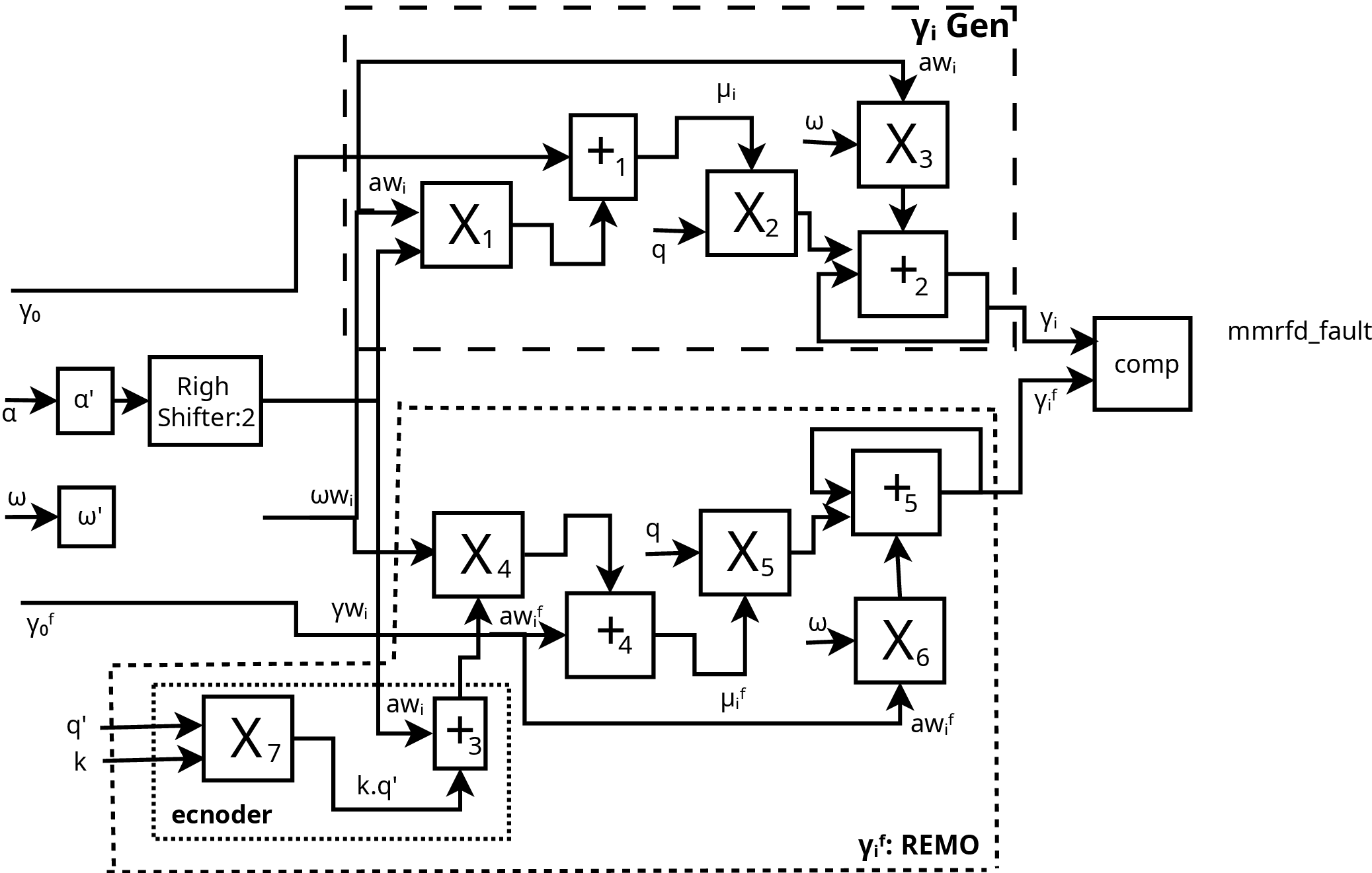}
\vspace{-5pt}
\caption{Hardware Architecture of REMO}
\vspace{-5pt}
\label{fig:remo}
\end{figure*}
\subsubsection{$\gamma_i$ $Gen$}
The $\gamma_i$ $Gen$ calculates lines \ref{line:mui} and \ref{line:gammai} of Algorithm \ref{algo:mmrfd}.
The $\times_1$, $\times_2$ and $\times_3$ in $\gamma_i$ $Gen$ are used to multiply $aw_i$ with $\omega'_0$, $\mu_i$ with $q$ and $\omega$ with $aw_i$ respectively. The $+_1$ and $+_2$ adders are used to compute $\mu_i$ and $\gamma_i$ respectively. If $p = 0$, the $+_2$ block provides the final output at the $\gamma_i$ register after $\frac{l}{w}$ cycles. If $p \neq 0$, the $+_2$ block provides the final output at the $\gamma_i$ register after $\frac{l}{w} + 1$ cycles.
\subsubsection{$REMO:$ $\gamma_i^f$}
The \texttt{REMO}: $\gamma_i^f$ calculates lines \ref{line:muif} and \ref{line:gammaif} of Algorithm \ref{algo:mmrfd}.
The $\times_4$, $\times_5$ and $\times_6$ in \texttt{REMO}: $\gamma_i^f$ are used to multiply $aw_i^f$ with $\omega'_0$, $\mu_i^f$ with $q$ and $\omega$ with $aw_i^f$ respectively. The additional multiplier and $\times_7$ is used to compute $k,q'$ in the encoder block. The $+_4$ and $+_5$ adders are used to compute $\mu_i^f$ and $\gamma_i^f$ respectively. One additional adder $+_3$ is used to calculate the encoded $aw_i^f$. If $p = 0$, the $+_5$ block provides the final output at the $\gamma_i^f$ register after $\frac{l}{w}$ cycles. If $p \neq 0$, the $+_2$ block provides the final output at the $\gamma_i^f$ register after $\frac{l}{w} + 1$ cycles. The comparator block ($comp$) compares $\gamma_i^f$ and $\gamma_i$. If they are not equal, the signal $mmrfd\_fault$ is asserted.
\section{Fault Detection in Memory Units}
\label{sec:mem}
NTT uses two types of memory units: RAMs and ROMs. RAMs are used to read and write polynomial coefficients during forward and inverse NTT operations. The constant terms in Learning With Errors (LWE)-based PQC algorithms, such as twiddle factors, are stored in ROM. Our fault detection method for the RAMs and ROMs of the NTT uses \texttt{Memory RC}, which is primarily constructed using \texttt{i-k RC} for RAMs and \texttt{i-j RC} for ROMs. These are presented in the next two sections based on the Kyber standard.
\subsection{Memory Address Protection in RAMs : \texttt{i-k RC}}
We have implemented Kyber-768, in which Forward NTT operations are applied during key generation, decryption, and encryption to the secret vector samples $s$, the ephemeral secret vector $r$, the error samples $e$, and the vector of polynomials $u$, which is the result of a matrix-vector multiplication with the public key matrix $A$. The Inverse NTT (INTT) operations are performed on $\hat{u}$, $\hat{us}$, $\hat{tr}$, $\hat{A^T} \circ \hat{r}$, and $\hat{t^T} \circ \hat{r'}$. Here $\hat{u}=NTT(Decompress(cipher))$, $\hat{us}=\hat{u}.\hat{s}$, $\hat{tr}=\hat{t}.\hat{r}$. $\hat{A^T}$ is the Generate matrix during encryption. For further details on Kyber, please refer to FIPS 203 \cite{fips203}. Our Kyber implementation uses one RAM for each polynomial, resulting in a total of 10 RAMs used for the key generation, decryption, and encryption processes. 

\begin{table}[!htbp]
    \centering
    \resizebox{\textwidth}{!}{%
\begin{tabular}{|c|c|
>{\centering\arraybackslash}p{1.2cm}|>{\centering\arraybackslash}p{1.2cm}|>{\centering\arraybackslash}p{1.2cm}|
>{\centering\arraybackslash}p{1.2cm}|>{\centering\arraybackslash}p{1.2cm}|>{\centering\arraybackslash}p{1.2cm}|
>{\centering\arraybackslash}p{1.2cm}|>{\centering\arraybackslash}p{1.2cm}|>{\centering\arraybackslash}p{1.2cm}|}
\hline
\multirow{2}{*}{\textbf{Block}} & \multirow{2}{*}{\textbf{NTTs}} 
& \multicolumn{3}{c|}{\textbf{Kyber-512 ($k=2$)}} 
& \multicolumn{3}{c|}{\textbf{Kyber-768 ($k=3$)}} 
& \multicolumn{3}{c|}{\textbf{Kyber-1024 ($k=4$)}} \\
\cline{3-11}
& & \# NTT Call & Total Hits on RAMs & Total ROM Hits &  \# NTT Call & Total Hits on RAMs & Total ROM Hits  &  \# NTT Call & Total Hits on RAMs & Total ROM Hits  \\
\hline\hline

\multirow{3}{*}{\textbf{KeyGen}} 
& NTT($s$) & 2 & 4096 &2048 & 3 &6144 & 3072& 4 &8192 &4096 \\
& NTT($e$) & 2 & 4096 & 2048 & 3 &6144 & 3072& 4 &8192 &4096 \\
& INTT & 0 & 0& 0& 0 &0 & 0& 0 & 0&0 \\
\hline
\multirow{3}{*}{\textbf{Encap}}
& NTT($r$) & 2 &4096 & 2048& 3 &6144 &3072 & 4 &8192 &4096 \\
& INTT($\hat{A^T} \circ \hat{r}$) & 2 &4096 &2048 & 3 &6144 & 3072& 4 & 8192&4096 \\
& INTT($\hat{t^T} \circ \hat{r}$) & 1 & 2048&1024 & 1 & 2048& 1024& 1 &2048 &1024 \\
\hline
\multirow{5}{*}{\textbf{Decap}}
& NTT($r$) & 2 & 4096& 2048& 3 &6144 & 3072& 4 &8192 &4096 \\
& NTT($u$) & 2 & 4096& 2048& 3 &6144 &3072 & 4 &8192 &4096 \\
& INTT($\hat{u}$) & 2 & 4096 & 2048& 3 &6144 & 3072& 4 &8192 &4096 \\
& INTT($\hat{us}$) & 1 &2048 &1024 & 1 &2048 & 1024& 1 & 2048&1024 \\
& INTT($\hat{tr}$) & 1 & 2048& 1024& 1 &2048 & 1024& 1 & 2048&1024 \\
\hline

\multicolumn{2}{|c|}{\textbf{total}}& 17  & \textbf{\textbf{34,816}}& \textbf{17,408}& 24 &\textbf{49,152} &\textbf{24,576} & 31 &\textbf{63,488} &\textbf{31,744 }\\

\hline
\end{tabular}
}
    \caption{Pattern of Memory hits in Kyber Variants using $i$, $j$ \& $k$}
    \label{tab:memory:hit}
\end{table}

In our Kyber, $i$ and $k$ indices are used to generate the addresses of 10 RAMs during read, write operations. In Table \ref{tab:memory:hit}, we present the number of RAM hits using the address bus $j$ during the key generation, decryption, and encryption processes for different variants of Kyber. The read and write operation on RAMs are considered as RAM hits. The data output of the RAMs is denoted by $\alpha$. 
The $i$ and $k$ patterns in each loop iteration during forward and inverse NTT is structured and hierarchical, where the upper bound $k$ are based on $i$. Therefore, we establish rules for each $CT$-$BU$ iteration: $k \leq s_i$, where $s_i$ is initialized to $n-1$ and is right-shifted by one bit every $256$ iterations. The $i-k$ RC takes $i$ and $k$ indices from $i-j-k~Gen$ and check the above mentioned rule. If the rule is not violated, the $ram~fault$ is set to $0$; otherwise, it is set to $1$.
\subsection{Memory Address Protection in ROM : \texttt{i-j RC}}
Kyber needs a ROM to store Twiddle factors ($\omega$), which required to read from ROM in each iteration of forward and inverse NTT. For more details about $\omega$, please refer to FIPS 203 \cite{fips203}. In our Kyber variants, the $i$ and $j$ patterns in each loop iteration during forward and inverse NTT are also structured and hierarchical, where the upper bound of $j$ depends on $i$, specifically $j \leq 2^i - 1$. The $i-j$ RC takes $i$ and $j$ indices from $i-j-k~Gen$ and check the above mentioned rule. If the rules is not violated, the $rom~fault$ is set to $0$; otherwise, it is set to $1$.

\section{Hardware Architecture of NTT with CT-BU, REMO\& Memory RC}
\label{sec:impl}
As shown in Fig. \ref{fig:ct_pipe}, the CT-BU unit has 3 pipeline stages: (i) For buffering $\omega$ and $U$ (ii) Calculation of $V$ and (iii) Update elements of point-wise representation. Our $CT-BU$ inside NTT iterates $log_2\times\frac{n}{2}-1$ times where $n-1$ is the degree of input polynomial. For our Kyber $n=256$. Therefore, the NTT in our Kyber iterates 1024 times. Line \ref{line:mmrfd} of Algorithm \ref{algo:intt} represents the $2^{nd}$ pipeline stage of the CT-BU, which computes the Montgomery modular reduction of $q$ on the product of $\omega$ and $\alpha[j+k+\frac{i}{2}]$.

\begin{figure}[!htb]
\centering
\vspace{-5pt}
\includegraphics[width=0.5\textwidth]{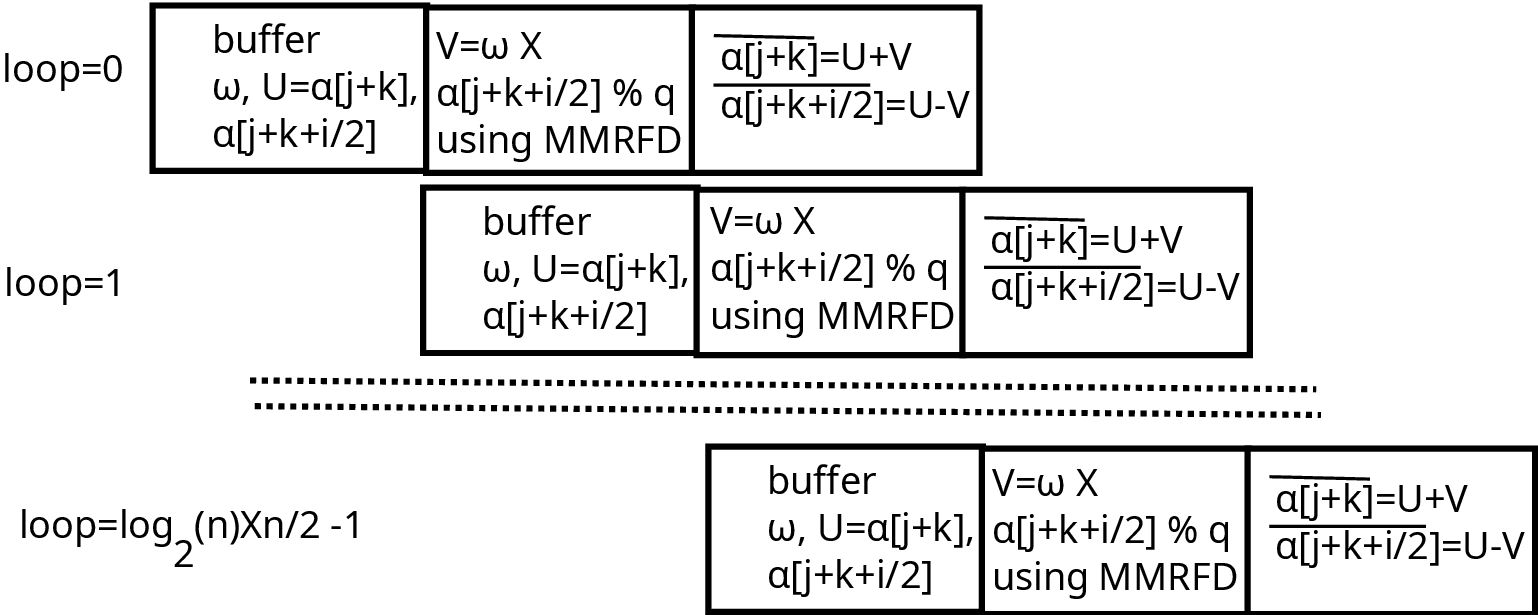}
\vspace{-5pt}
\caption{Pipeline Stages of CT-BU}
\vspace{-5pt}
\label{fig:ct_pipe}
\end{figure}

As shown in Fig. \ref{fig:ct_bf_arch}, the polynomial coefficients and twiddle factors $\omega$ are stored in RAMs and ROM, which can be accessed by various hardware blocks such as the polynomial multiplier, polynomial adder, NTT and etc. The specific hardware block that accesses the memory depends on the requirements of the PQC algorithm. The $addr$, $rd\_en$,  $wr\_en$ and $din$ of the memory block can be accessed by different hardware blocks by selecting inputs of the muxes controlled by $Control~Unit$. Similarly using $demux$, different hardware blocks can read  $\alpha$ from RAMs through $dout$. Then, $\alpha$ is connected to the $CT-BU$ for the $U$ and $V$ calculation (line \ref{line:u} and \ref{line:mmrfd} of Algorithm \ref{algo:intt}). The calculation of $V$ depends on the $MMRFD$ block which follows Algorithm \ref{algo:mmrfd} to multiply $\alpha$ and $\omega$ in a $w$ word-wise fashion. The $\gamma$ block shown in Fig. \ref{fig:ct_bf_arch}, takes $w$ bits named $aw_i$ at a time from the $\alpha$.  The $\gamma^f$ block shown in Fig. \ref{fig:ct_bf_arch}, takes $w$ bits named $aw_i^f$ at a time from the $\alpha$. The $encoder$ block calculates $aw_i^f$ (line \ref{line:awif} of Algorithm \ref{algo:mmrfd}) and sends it to the $\gamma^f$ block. The $\gamma$ block executes $u_i$ and $\gamma_i$ (line \ref{line:mui} and line \ref{line:gammai} of Algorithm \ref{algo:mmrfd}). The $\gamma^f$ block executes $u_i^f$ and $\gamma_i^f$ (line \ref{line:muif} and line \ref{line:gammaif} of Algorithm \ref{algo:mmrfd}). The modules $2^w$ operation to calculate $u_i$ and $u_i^f$ in the $\gamma$ block and $\gamma^f$ block, respectively, are performed by restricting the size of $u_i$ and $u_i^f$ to $w$ bits.  The division by $2^w$ operation to calculate $\gamma_i$ and $\gamma_i^f$ in the $\gamma$ block and $\gamma^f$ block, respectively, are performed by removing $w$ bits form the right side of the size of $\gamma_i$ and $\gamma_i^f$. These two bitwise operations replace resource- and latency-intensive modulus and division operations, significantly reducing slice usage and delay.
The $comp$ block compares $\gamma$ and $\gamma^f$. If both values match, $mmrfd\_fault=0$ and $V$ is used to calculate $\alpha[j+k]$ and $\alpha[j+k+\frac{i}{2}]$. Otherwise, $mmrfd\_fault=1$ is initiated. The $\alpha[j+k]$ and $\alpha[j+k+\frac{i}{2}]$ are calculated by the $adder$ and $sub$ blocks respectively.
\begin{figure}[!htb]
\centering
\includegraphics[width=0.75\textwidth]{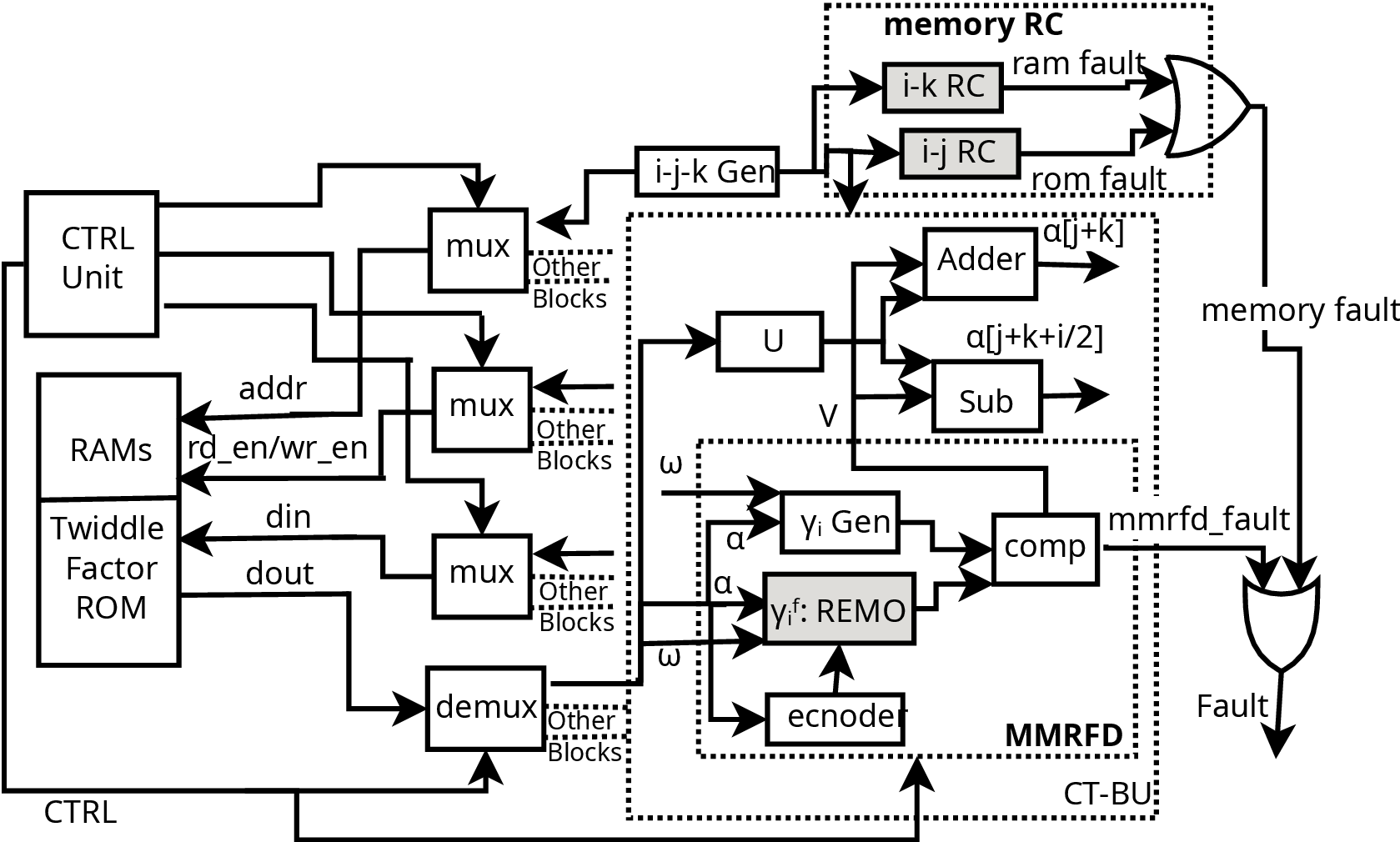}
\caption{NTT Architecture With CT-BU, REMO \& Memory RC}
\label{fig:ct_bf_arch}
\end{figure}

\begin{table}[htbp!]
\centering
\begin{tabular}{
|
>{\centering\arraybackslash}p{4cm}|
>{\centering\arraybackslash}p{1.5cm}|
>{\centering\arraybackslash}p{1.5cm}|
>{\centering\arraybackslash}p{1.5cm}|
>{\centering\arraybackslash}p{1.5cm}|
>{\centering\arraybackslash}p{1.5cm}|
>{\centering\arraybackslash}p{1.5cm}|
>{\centering\arraybackslash}p{1.5cm}|
}
\hline
\textbf{Architecture} & \textbf{n, q} & \textbf{SECs}& \textbf{Slices} & \textbf{LUTs} & \textbf{FFs} &\textbf{DSPs/ BRAMS}& \textbf{Power (mW)}  \\ \hline\hline
								
				\textbf{Kyber Montgomery (Baseline)} & 256,                   & 173 &73 & 242 & 100 &1/0 &104  \\ \cline{1-1} \cline{3-8}
				\textbf{Kyber Montgomery (Protected)} &3329& 287&89 & 275 & 139 & 2/0&107  \\
				\hline \hline
				\textbf{CRYSTALS-Dilithium Montgomery (Baseline)}   &256,   &  150    &50 &121  &111 &1/0&100  \\ \cline{1-1} \cline{3-8}
				\textbf{CRYSTALS-Dilithium Montgomery (Protected)} & 8380417&274  &74&171 &148&2/0&102  \\ \hline\hline	
				
				\textbf{Falcon Montgomery (Baseline)}               & 512,          &136&36 &84&70&1/0&99   \\ \cline{1-1} \cline{3-8}
				\textbf{Falcon Montgomery (Protected)} & 12289 &248&48&128&103&2/0&101  \\ \hline\hline	
				
				\textbf{NTRU Montgomery (Baseline)}                   & 2048,   &132& 32 & 84 & 78 & 1/0&106 \\ \cline{1-1} \cline{3-8}
				\textbf{NTRU Montgomery (Protected)} &12289  &249& 49 & 122 & 120 & 2/0&110 \\ \hline\hline
				
				\multicolumn{8}{|c|}{\textbf{Artix-7 (xc7a100tcsg324-3), w=4, clock=100MHz}} \\ \hline
\end{tabular}
			\vspace{2pt}
			\caption{Overhead of REMO \& Memory RC in Different PQC Algorithms}
			\label{tab:compalgo}
\end{table}

\section{Results \& Discussions}
\label{sec:res}
This section first covers the fault detection scheme's overheads, then its coverage.
\subsection{Overheads}
The design is implemented on an $Artix$-$7$ $(xc7a100tcsg324$-$3)$ FPGA using the $Vivado$ $22.02$ tool and the VHDL language. The proposed fault detection model is implemented in Kyber, Crystal Dilithium, Falcon and NTRU algorithms.The implementation costs of the proposed Fault Detection ($FD$) model, which includes both \texttt{REMO} and \texttt{Memory RC}, for the Kyber, CRYSTALS-Dilithium, Falcon, and NTRU algorithms are presented in Table \ref{tab:compalgo}. These costs are measured in terms of slices, LUTs, flip-flops, DSPs, and power consumption. Additionally, we calculate the Slice Equivalent Cost (SEC) as defined in Equ. \ref{equ:sec}, following the method described in \cite{liu}.
\begin{equation}
\label{equ:sec}
SEC=Slices+ DSPs \times  100+ BRAMS \times 200    
\end{equation}
Our \texttt{REMO} model for the Kyber standard, with $n=256$ and $q=3329$, utilizes only 8 slices (comprising 15 LUTs and 8 Flip-Flops) and a single DSP block, as shown in Table \ref{tab:comp}. Additionally, \texttt{Memory RC} consumes 8 more slices.
Table \ref{tab:lit} shows a comparison of error coverage and overheads in terms of area, delay, and energy with existing fault detection literature.
The number of extra slices and DSP required for $MMR$ to detect faults using \texttt{REMO} are $8$ and $1$ respectively. 
  The \texttt{Memory RC} also required 8 slices. Therefore, the total Area Overhead ($AO$) of fault detection of $NTT$ is shown in Equ. \ref{equ:ao}.
  \begin{align}
 \label {equ:ao}
 AO=& \frac{SEC~of~REMO+SEC~of ~memory~RC}{SEC~of~CT-BU}\times 100\\
 &=\frac{108+8}{1356} \times 100=8.5\% \nonumber
  \end{align}
  \par 
  As shown in Table \ref{tab:lit}, we compare our fault detection model with other fault detection solutions designed for NTT and other different computational units of cryptographic algorithms. While the solutions different computational units of cryptographic algorithms are not directly comparable to our proposed fault detection method, this comparison provides an overview of the implementation cost required to achieve a given level of error coverage.
As shown in Table \ref{tab:lit}, this fault detection module has a slice overhead of $8.5\%$ compared to the unprotected NTT. The fault detection hardware consumes only $3\text{mW}$ of power, resulting in a $1.8\%$ power overhead compared to the unprotected NTT. The proposed fault detection method runs in parallel with the NTT component. Adopting this proposed fault detection logic into the NTT does not impact the critical path, clock period, or the number of clock cycles required for the baseline NTT. Consequently, the implementation incurs a $0\%$ delay, $8.5\%$ slice overhead and $1.8\%$ energy overhead, which is highly reasonable and competitive with the existing literature for detecting $87.2\%$ to $100\%$ fault occurrences. It is to be noted that the our fault detection unit operates with a delayed clock compared to the main $NTT$. This is done to enable the detection of both transient and permanent faults.

\begin{figure}[!htb]
\centering
\includegraphics[width=0.5\textwidth]{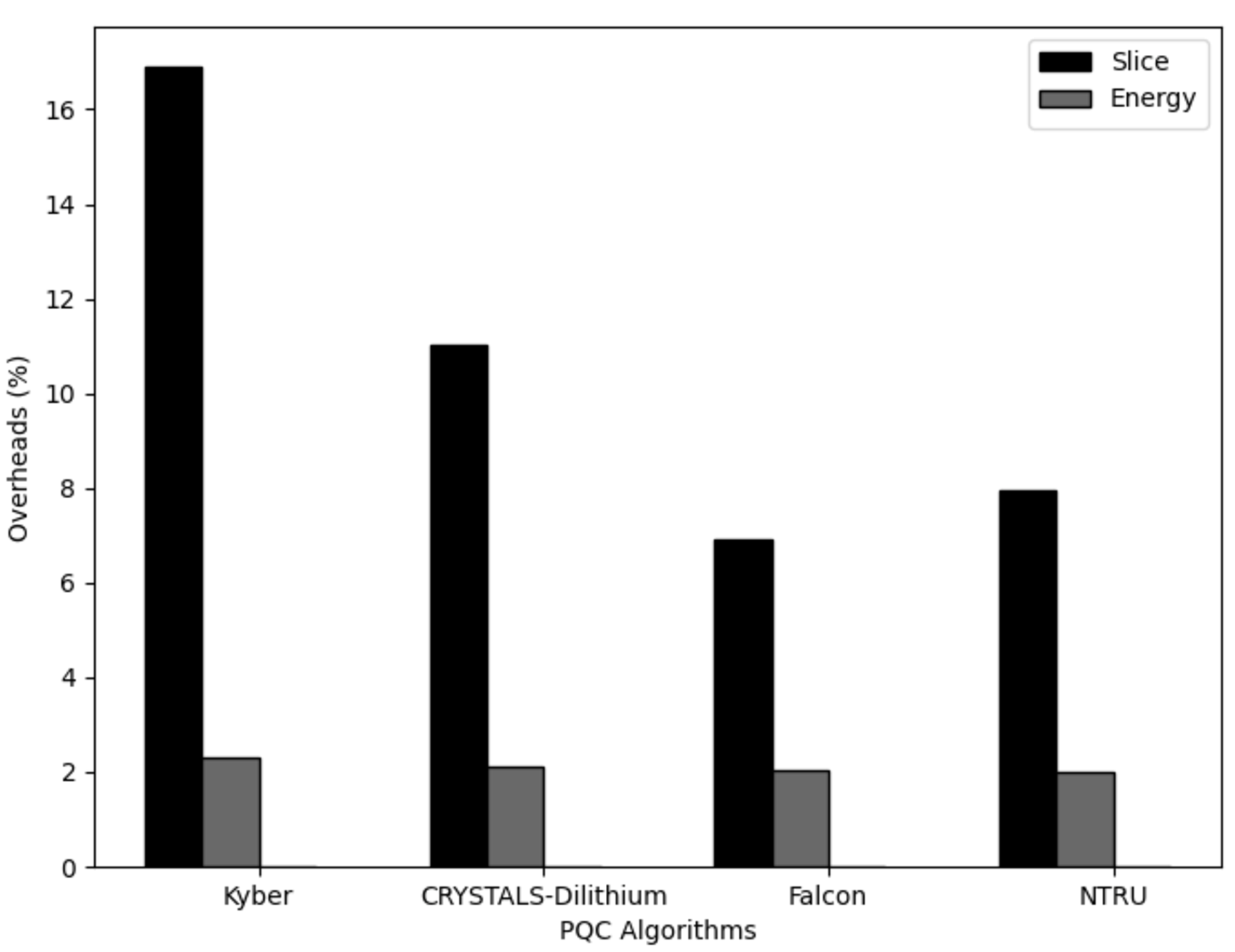}
\vspace{-5pt}
\caption{Overhead of REMO  \& Memory RC for Different PQC Algorithms}
\label{fig:graph}
\end{figure}

\begin{table*}[ht]
    \centering
    \begin{tabular}{|>{\centering\arraybackslash}p{1.4cm}|
                    >{\centering\arraybackslash}p{0.4cm}|
                    >{\centering\arraybackslash}p{0.7cm}|
                    >{\centering\arraybackslash}p{0.7cm}|
                    >{\centering\arraybackslash}p{0.7cm}|
                    >{\centering\arraybackslash}p{0.7cm}|
                    >{\centering\arraybackslash}p{0.8cm}|
                    >{\centering\arraybackslash}p{0.8cm}|
                    >{\centering\arraybackslash}p{1.15cm}|}
    \hline
    \textbf{Block Names} & \textbf{SEC} & \textbf{Slices} & \textbf{LUTs} & \textbf{FFs} & \textbf{DSPs} & \textbf{BRAMs} & \textbf{Power (mW)} & \textbf{Critical Path (ns)} \\
    \hline
    Kyber-768 (baseline) & 3395 & 1795 & 6008 & 4404 & 2 & 7 & 425 & 9.79 \\
    \hline
    NTT/INTT (baseline) & 1356 & 556 & 1711 & 1204 & 2 & 3 & 163 & 9.39 \\
    \hline
    CT BU (baseline) & 685 & 285 & 778 & 207 & 2 & 1 & 127 & 9.39 \\
    \hline
    MMR (baseline) & 173 & 73 & 242 & 100 & 1 & 0 & 104 & 9.87 \\
    \hline
    \textbf{\texttt{REMO}} & 108 & 8 & 15 & 8 & 1 & 0 & 2 & 9.11 \\
    \hline
    \textbf{\texttt{Memory RC}} & 8 & 8 & 18 & 31 & 0 & 0 & 1 & 7.64 \\
    \hline
    \multicolumn{9}{|c|}{\textbf{Artix-7 (xc7a100tcsg324-3), n=256, q=3329, l=12, w=4, clock=100MHz}} \\
    \hline
    \end{tabular}
    \caption{Overhead of REMO \& Memory RC for Kyber Standard}
    \label{tab:comp}
\end{table*}

 	 \begin{table}[!htbp]
 		\centering
	\begin{tabular}{|p{0.5cm}|>{\centering\arraybackslash}p{2.7cm}|>{\centering\arraybackslash}p{2.5cm}|>{\centering\arraybackslash}p{1.8cm}|>{\centering\arraybackslash}p{1.8cm}|>{\centering\arraybackslash}p{1.6cm}|>{\centering\arraybackslash}p{1.6cm}|>{\centering\arraybackslash}p{1.7cm}|}
	\hline
	\textbf{Work}& \textbf{Type of Fault } & & \textbf{Baseline} &\multicolumn{3}{c|}{\textbf{Overhead (\%)}}& \textbf{(\%) Error} \\ \cline{5-7}
	
	& \textbf{Detection} & \textbf{Target HW} & \textbf{(for Overhead)} &\textbf{Area} & \textbf{Delay} & \textbf{Energy} &  \textbf{Coverage} \\ \hline

	\cite{canto}  &CRC5 & sub, add of McEliece crypto &McEliece crypto & 18.33   & 11.25   & $\sim$0  & $>$99.9 \\ \hline

	\cite{saeed}  &REMO & $x^y~mod~n$ &$x^y~mod~n$& 0.8   & 0.27   & 0.65  & 97.1-100 \\ \hline

	\cite{howe}*  &look up table based output distribution check low cost/ standard/ expensive & CDT Error Sampler &CDT Error Sampler & 9.09/ 77.4/ 18.2    & NR   & NR & NR\\ \hline	
	\cite{canto2}  &RESO & MAC unit of Saber/ NTRU/ FrodoKEM & Saber/ NTRU/ FrodoKEM & 36.6/39.6/ 28.4  & 28.3/16.7/ 32.7   & 1.2/3.2/ $\sim$0 & $>$99.9 \\ \hline
	\cite{kermani}  &Recomputing with swapped ciphertext & Galois Counter Mode & AES-GCM & 4.9/6.7   & NR   & NR  & 100 \\ \hline	
	
	\cite{kamal}  &Spatial duplication & NTRU&NTRU& 6.22   & NR   & NR  & 100 \\ \hline
	\cite{cintas}  &1/2/3-bit parity& Goppa Arithmetic &McEliece& 9.8/11.3/9.6   & 1.4/0.8/1   & 2.7/2.7/2.7  & 100 \\ \hline\hline
	 \cite{jati}  &Randomized Memory Addr. & NTT &Kyber& NR   & NR   & NR  & NR\\\hline
	  \cite{rafael}  & Local Masked & NTT &NTT& NR   & NR   & NR  & NR\\\hline
	  \cite{ravi2}  & Randomized Memory Addr. + Local Masked & NTT &NTT& NR   & NR   & NR  & NR\\\hline
	 \cite{khan}  &Hamming A-O/ R-O & NTT Memory &Kyber& 16.4/ 19.2   & 0/ 0   & NR  & NR\\\hline
  \cite{khan}  &Hamming+ Parity A-O/ R-O & NTT Memory &Kyber& 10.8/ 21.5   & 1.6/ 1.44   & NR  & NR\\ \hline
   \cite{sven}  &Polynomial Evaluation and Interpolation & CT-BU &Dilithium& NR   & 72   &  NR  & NR\\ \hline
	\cite{sarker}  &RENO Spartan7 v1/v2/v3 &Butterfly Unit& NTT& 20.2/15.3/21.5   & 8.46/15.88/13.71   & 15.6/7.6/11.2  & 99.51/99.67\ /99.41 \\ \hline
        \cite{sarker}  &RENO Zynq v1/v2/v3 & Butterfly Unit&NTT& 24/7.5/17  & 9.32/19.66/ 21.78   & 20.47/13.27/ 17.26  & 99.51/99.67 /99.41 \\ \hline
        
\multirow{3}{*}{\textbf{Our*}} 
& \multirow{3}{*}{\textbf{REMO + Memory RC}} 
& \multirow{3}{*}{\textbf{CT-BU \& Memories}} 
& \textbf{CT-BU \& Memories} & \textbf{16.9} & \textbf{0} & \textbf{2.3} & \multirow{3}{*}{\textbf{\makecell{87.2-100 \\(REMO),\\50.7-100 \\(Memory RC)}}} \\ \cline{4-7}
& & & \textbf{NTT} & \textbf{8.5} & \textbf{0} & \textbf{1.8} & \\ \cline{4-7}
& & & \textbf{Kyber-768} & \textbf{3.18} & \textbf{0} & \textbf{0.7} & \\ \hline
\end{tabular}
\vspace{2pt}
\\Note: NR = Not Reported; $*$ =Eq. \ref{equ:sec} is used to calculate area overhead.\\
 		\caption{Overhead Comparison with literature}
 	\label{tab:lit}
 \end{table} 
\subsection{Error Coverage}
To measure the error coverage of the proposed \texttt{REMO} and \texttt{Memory RC} schemes, we simulated the fault injection process using Python on a $i5$ processor with $8$ GB of RAM, running Ubuntu $24.04$. Both the simulation process uses two fault modes: random faults and burst faults. The random mode flips $\eta$ number of bits randomly, whereas the burst mode flips $\eta$ number of consecutive bits.   In both fault modes, bit flips mean that '$0$'(s) are turned into '$1$'(s) and '$1$'(s) are turned into '$0$'(s).
\subsection{Error Coverage of REMO}
This simulation method utilized $1.5$ million samples.
The random and burst fault modes inject faulty bits into $\alpha$, $\omega$, and both $\alpha$ and $\omega$. Table \ref{tab:fault:remo} shows that the fault detection efficiency varies from 87.2\% to 100\%, depending on the word size $w$, the number of faulty bits $\eta$ and the fault mode.
From Table \ref{tab:fault:remo} three conclusions can be drawn: 
 \begin{itemize}
  \item The size of $w$ has minimal impact on fault detection efficiency but significantly affects the area, as measured by SECs.
  \item As the number of faulty bits increases, the fault detection efficiency improves. 
  \item The fault detection efficiency is higher in random mode compared to burst mode. 
  \end{itemize}

  \begin{table}[htbp!]
	\centering
	\begin{tabular}{|>{\centering\arraybackslash}p{0.3cm}|>{\centering\arraybackslash}p{2.0cm}|>{\centering\arraybackslash}p{1.5cm}|>{\centering\arraybackslash}p{1.5cm}|>{\centering\arraybackslash}p{1.5cm}|>{\centering\arraybackslash}p{1.5cm}|>{\centering\arraybackslash}p{1.5cm}|>{\centering\arraybackslash}p{1.5cm}|>{\centering\arraybackslash}p{1.5cm}|}
		\hline
		\textbf{w}&\textbf{\#} & \multicolumn{6}{c|}{\textbf{Fault Detection Efficiency(\%)}}&\textbf{SECs}\\ \cline{3-8}
		& \textbf{fault}  & \multicolumn{2}{c|}{ \textbf{fault in $\alpha$}} & \multicolumn{2}{c|}{\textbf{fault in $\omega$}} & \multicolumn{2}{c|}{\textbf{fault  in $\alpha$ \& $\omega$}}& \textbf{of REMO}\\ \cline{3-8}
		       & \textbf{bits($\eta$)}   & \textbf{random}&\textbf{burst} &\textbf{random}&\textbf{burst}   &\textbf{random}&\textbf{burst} &\textbf{ (Eq. \ref{equ:sec})} \\ \hline
		       	& 1   &  87.24&- & 98.33& -  &98.89& -&  \\
			& 3   &  96.72&93.67 &100 &99.99  &100&100&   \\
			& 5   &99.03  &96.83 &100 &100  &100&100 &  \\
		2	& 11   &99.95  &99.61 &100 &99.99  &100&100 &103  \\
			& 17   &99.99  &99.95 & 100&100  &100&100 &  \\
			& 23   &99.99  & 99.99& 100&100  &100& 100&  \\\hline
			\hline
		       & 1   &  87.21&- &98.26& -   &98.89&   -&  \\ 
		       & 3   &   95.94&90.34&  99.01& 99.97 & 100&  100&   \\ 
			& 5   &  98.3& 93.68&   100&100 &  100&100&  \\ 
		4	& 11  & 99.68&97.63  &100 &99.99   &100&100 & 108   \\ 
			& 17  & 99.87&  99.27&100&100    &100&100&     \\ 
			& 23  & 99.98& 99.79 & 100& 100  &100& 100 &   \\ 
			\hline\hline
			& 1   &  87.2&- &97.65&-  &98.87& - &\\ 
			& 3   &  94.43&88.51 &99.99&99.81    &100& 100&  \\ 
			& 5   &  96.55&89.98 &100 &99.99   &100 &  100&\\
		8	& 11   &  98.19 &94.28&100&100    &100&  100&179 \\
			& 17   &  99.03&97.1 &100 &99.99   &100 & 100& \\
			& 23   &  99.85 &98.33&100 &99.99   &100& 100 & \\
		\hline
\multicolumn{9}{|c|}{\textbf{l=24, sample size=1.5 million}}\\\hline		
		
	\end{tabular}
	\caption{Error Detecting Efficient for $\eta$ bit Random \& Burst Flipping using REMO}
	\label{tab:fault:remo}
\end{table}

\subsection{Error Coverage of Memory RC}
To measure the error coverage of the \texttt{Memory RC}, we executed the key generation, decryption and encryption processes of Kyber-768 a total of $300$ times. Each execution of Kyber-768 results in $49,152$ RAM accesses (as shown in Table \ref{tab:memory:hit}) and $24,576$ ROM accesses to read the twiddle factors. Among these 300 executions, faults were injected using three different modes, with 100 executions allocated to each mode: (i) faults in both RAM and ROM, (ii) faults only in ROM, and (iii) faults only in RAM. Table \ref{tab:fault:memory} shows that the fault detection efficiency varies from 50.7\% to 100\%, depending on the number of faulty bits $\eta$ and the fault mode.
\begin{table}[!htbp]
\centering
\caption{Error Detecting Efficient for $\eta$ bit Random \& Burst Flipping using Memory RC for Kyber Standard}
\resizebox{\textwidth}{!}{%
\begin{tabular}{|c|cc|cc|cc|}
\hline
\multirow{3}{*}{\textbf{\# Faulty Bits ($\eta$)}} 
& \multicolumn{6}{c|}{\textbf{Fault Detection Efficiency(\%)}} \\\cline{2-7}

& \multicolumn{2}{c|}{\textbf{Faulty addr in RAMs and ROM (j \& k)}} 
& \multicolumn{2}{c|}{\textbf{Faulty addr in ROM  (j only)}} 
& \multicolumn{2}{c|}{\textbf{Faulty addr in RAMs  (k only)}} \\
\cline{2-7}
 & \textbf{Random (\%)} & \textbf{Burst (\%)} 
 & \textbf{Random (\%)} & \textbf{Burst (\%)} 
 & \textbf{Random (\%)} & \textbf{Burst (\%)} \\
\hline
1 & 87.79 & -     & 53.63 & -     & 50.07 & -     \\
2 & 98.57 & 93.78 & 67.53 & 58.13 & 66.53 & 56.12 \\
3 & 99.93 & 98.16 & 75.33 & 63.54 & 75.01 & 62.35 \\
4 & 100   & 100   & 80.27 & 69.34 & 79.99 & 68.76 \\
5 & 100   & 100   & 83.41 & 75.37 & 83.41 & 75.11 \\
6 & 100   & 100   & 85.81 & 81.44 & 85.81 & 81.24 \\
7 & 100   & 100   & 87.5  & 87.5  & 87.5  & 87.5  \\
\hline
\end{tabular}%
}
\label{tab:fault:memory}
\end{table}

\section{Conclusion}
\label{sec:con}
In this manuscript, we present a light wight modified Montgomery reduction integrated within a CT-BU for fault detection, capable of addressing both permanent and transient faults. It uses the \texttt{REMO} method. The results demonstrate that our fault detection scheme achieves a high fault detection rate with minimal resource and power overhead, without affecting the critical path of the original design. Although the fault detection scheme proposed in this paper is specifically designed for Montgomery reduction within the CT-BU, it can also be applied to any hardware implementing polynomial multiplication with modular reduction, where Montgomery reduction is utilized. The \texttt{Memory RC} can detect between 50.7\% and 100\% of faults in the memory units used in the NTT. To the best of our knowledge, our fault detection method has one of the lowest slice overheads among existing fault-tolerant techniques in the literature on PQC. It is important to note that we have not explored fault detection methods for the contents of the memory units, as several efficient techniques such as hamming codes, parity bits, and CRC are already well-established in the literature \cite{khan} for protecting memory contents. The code of this work is uploaded to GitHub \footnote{https://github.com/rourabpaul1986/NTT}.\\
\textbf{Acknowledgment}
This publication has emanated from research conducted with the financial support of Taighde Éireann - Research Ireland under Grant number 13/RC/2077\_P2 at CONNECT: the Research Ireland Centre for Future Networks.

\bibliographystyle{unsrt}  
\bibliography{example}
\nocite{*}

\end{document}